\documentclass[runningheads,a4paper]{llncs}
\usepackage{amssymb}
\usepackage{graphicx}
\usepackage{color,xcolor}
\usepackage{amsmath,mathrsfs}
\usepackage{hyperref}

\usepackage{url}
\urldef{\mailsa}\path|{alfred.hofmann, ursula.barth, ingrid.haas, frank.holzwarth,|
\urldef{\mailsb}\path|anna.kramer, leonie.kunz, christine.reiss, nicole.sator,|
\urldef{\mailsc}\path|erika.siebert-cole, peter.strasser, lncs}@springer.com|

\makeatletter

\newcommand{\Rmnum}[1]{\expandafter\@slowromancap\romannumeral #1@}
\makeatother

\newcommand{\tr}[2][1]{Tr_{#1}^{#2}}

\newcommand{\GF}[2][2]{{\mathbb F}_{#1^{#2}}}

\newcommand{\V}[2][2]{{\mathbb F}_{#1}^{#2}}
\newcommand{\B}{\mathcal B}
\newcommand{\F}{\mathbb F_2}
\newcommand{\AI}{{AI}}
\newcommand{\FAI}{{FAI}}
\newcommand{\FAA}{\mathcal{FAI}}
\newcommand{\LDA}{{LDA}}
\newcommand{\AN}{{AN}}

\newcommand{\wt}{wt}
\newcommand{\mindeg}{mindeg}

\newcommand{\MUL}{{MUL}}

\begin{document}
\title{Fast algebraic immunity of Boolean functions and LCD codes}
\titlerunning{Fast algebraic immunity of Boolean functions and LCD codes}

\author{Sihem Mesnager \inst{1}
\and Chunming Tang \inst{2}
}

\authorrunning{S. Mesnager and C. Tang}

\institute{Department of Mathematics, University of Paris VIII, 93526 Saint-Denis, University Sorbonne Paris Cit\'e, Laboratory Analysis, Geometry and Applications (LAGA), UMR 7539, CNRS, 93430 Villetaneuse,  and Telecom Paris, 91120 Palaiseau, France\\
\email{smesnager@univ-paris8.fr}\\
\and
School of Mathematics and Information, China West Normal University\\
 Nanchong, Sichuan,  637002, China\\
\email{tangchunmingmath@163.com}}

\maketitle

\begin{abstract}
Nowadays, the resistance against algebraic attacks and fast algebraic
attacks are considered as an important cryptographic property for
Boolean functions used in stream ciphers. Both attacks are very
powerful analysis concepts and can be applied to symmetric
cryptographic algorithms used in stream ciphers.
The notion of algebraic immunity has received wide attention since
it is a powerful tool to measure the resistance of a Boolean function
to standard algebraic attacks. Nevertheless, an algebraic tool to
handle the resistance to fast algebraic attacks is not clearly
identified in the literature. In the current paper, we propose a new
parameter to measure the resistance of a Boolean function to fast
algebraic attack. We also introduce the notion of fast immunity profile and
show that it informs both on the resistance to standard and fast
algebraic attacks. Further, we evaluate our parameter for two
secondary constructions of Boolean functions.
Moreover, A coding-theory approach to the characterization of perfect algebraic immune functions is presented.
Via this characterization, infinite families of binary linear complementary dual codes (or LCD codes for short) are obtained from perfect algebraic immune functions.
The binary LCD codes presented in this paper have applications in armoring implementations against
 so-called side-channel attacks (SCA) and fault non-invasive attacks, in addition to their applications in communication and data storage systems.

  \end{abstract}

{\bf Keywords} Boolean function $\cdot$ (Fast) Algebraic immunity $\cdot$ Algebraic attack  $\cdot$ Fast algebraic attack  $\cdot$ Reed-Muller code $\cdot$ LCD code $\cdot$ Side-channel attack $\cdot$ Fault injection attack.

\section{Introduction}
Boolean functions have important applications in the combiner model and the filter model of stream ciphers. A function used in such an application should mainly possess balancedness, a high algebraic degree, a high nonlinearity and, in the case of the combiner model, a high correlation immunity. In 2003, new kinds of attacks drawn from an original idea of
Shannon \cite{Shannon1949}  emerged; these attacks are called
\emph{algebraic attacks} and \emph{fast algebraic attacks}
\cite{Courtois2003,CourtoisMeier2003,MeierPasalicCarlet2004}.  Since 2003, the designers of cryptosystems in symmetric cryptography need also to ensure resistance to the algebraic attack (they need in practice optimal or almost optimal algebraic immunity) and good resistance to fast algebraic attacks and to the R{\o}njom-Helleseth attack \cite{RH2007,RH2011}, and its improvements.  A first nice primary construction of an infinite class of functions satisfying all the cryptographic criteria (balancedness, the algebraic degree the algebraic immunity and the non-linearity)  is the so-called Carlet-Feng construction \cite{CarletFeng2008}. Note that its good resistance to fast algebraic attacks has first been checked by computer for $n\leq 12$, using an algorithm from \cite{ACGKMR2006}, and later shown mathematically in \cite{Liu2012} for all n. Later, classes of functions have been proposed in the literature in which the authors suggested some modifications of the Carlet-Feng functions and other constructions (see for instance \cite{Carlet2013,CarletTang2015} and the references therein).

(Fast-) algebraic attacks have changed the situation in symmetric cryptography for the steam ciphers by
adding a new criterion of considerable importance to the above list. They
proceed by modeling the problem of recovering the secret key through an over-defined system of multivariate nonlinear equations of
algebraic degree at most $\deg(f)$. The core of algebraic attacks is
to find out low degree Boolean functions $g\not=0$ and $h$ such that
$fg=h$. It is shown in \cite{MeierPasalicCarlet2004} that this is
equivalent to the existence of low algebraic degree \emph{annihilators} of $f$,
that is, of $n$-variable Boolean functions $g$ such that
$f\cdot g = 0$ or $(1+ f)\cdot g=0$. The minimum degree of such $g$ is
called the \emph{algebraic immunity} of $f$, and we denote it by
$\AI(f)$. It must be as high as possible (the optimum value of
$\AI(f)$ being equal to $\left\lceil\frac{n}{2}\right\rceil$).
In 2020, a  novel application of Boolean functions with high algebraic immunity in minimal codes has been derived in \cite{CDMT2020}.
Fast algebraic attacks proceed differently and exploit the existence
of function $g$ of small degree such that the degree of $f\cdot g$ is
not too large. Many authors have indicated that having a high
algebraic immunity is not only a necessary condition for resistance to
standard algebraic attacks but also for resistance to fast algebraic
attacks. Nevertheless, having a high algebraic immunity may not be
sufficient in the design of pseudo-random generators using a Boolean function as filter or combiner
 (see \cite{Carlet2010}).  That motivates to define a new parameter to measure
the resistance of the Boolean function $f$ used in such generators to fast
algebraic attacks. Such a parameter has been proposed in
\cite{CarletTang2015,DuZhangLiu2012,589507}.
Very recently, M\'eaux has studied in \cite{Meaux_FAI_Maj,Meaux_FAI_Thres} the fast algebraic immunity of interesting families of cryptographic Boolean functions, namely the so-called majority functions (which have been intensively studied in the area of cryptography, because of their practical advantages and good properties), and Threshold functions (which are a sub-family of symmetric Boolean functions, which means that the output is independent of the order of the input binary variables). In 2020, Tang \cite{Tang2020} has derived a relation on the fast algebraic immunity between a Boolean function and its modifications, by introducing a new concept called partial fast algebraic immunity. As applications of this relation, he derived some upper bounds on the fast algebraic immunity of several known classes of modified majority functions with optimal algebraic immunity. These bounds show that these modified majority functions still have low fast algebraic immunity, which is coincident with the relation. A very nice reference on this topic is the excellent book of Carlet \cite{Book-Carlet} (which will appear soon).

In this paper, we provide for the first time a link between fast algebraic immune Boolean functions and the so-called \emph{linear complementary dual  code} (abbreviated LCD). An LCD code  is defined as a linear code $\mathcal C$ whose (Euclidean) dual code $\mathcal C ^ \perp$ satisfies $\mathcal C \cap \mathcal C^ \perp=\{\mathbf{0}\}$. LCD codes have been widely applied in data storage, communications systems, consumer electronics, and cryptography. In \cite{mas92}, Massey showed that LCD codes provide an optimum linear coding solution for the two-user binary adder channel. In 2014, Carlet and Guilley  \cite{CG14} investigated an interesting application of binary LCD codes against side-channel attacks (SCA) and fault injection attacks (FIA) and presented several constructions of LCD codes. It was shown non-binary LCD codes in characteristic $2$ can be transformed into binary LCD codes by expansion.  It is then important to keep in mind that, for SCA, the most interesting case is when the code is defined over an alphabet of size $q$ with $q$ even. The recent literature is abundant about LCD codes. One of the most important results on the classification of LCD codes is that any linear code over $\mathbb F_{q}$ ($q>3$) is equivalent to an (Euclidean) LCD code \cite{CMTQP2018}. A complete state-of-the-art on LCD codes can be found in the recent article \cite{CLM2019} and the references therein.

 This paper is organized as follows. In Subsection
\ref{sec:new-definition-fast}, we modify the parameter proposed in
\cite{DuZhangLiu2012,589507} so that it does not depend on the
algebraic immunity as in \cite{DuZhangLiu2012,589507} that we denote
by $\FAI (f)$. We show that the value of the modified parameter is less or
equal to the one proposed in \cite{DuZhangLiu2012} and \cite{589507} for every
Boolean function. 
In Subsection \ref{sec:fast-immun-prof}, we introduce the notion of
the immunity profile of a Boolean function and show that both the
algebraic immunity and the $\FAI$ of a Boolean function can be deduced
from this immunity profile. 
In Subsection \ref{sec:fast-immun-algebr}, we show that if a function is at low Hamming distance from a low algebraic degree function, then it is weak against fast algebraic attacks and we study further the behavior of $\FAI$. 
In Subsection \ref{sec:fast-algebr-immun}, we study
the $\FAI$ of a classical secondary construction of Boolean function,
which it called concatenation Boolean function. We prove that the
$\FAI$ of the concatenation of the Boolean function can be bounded from
below and above by the $\FAI$ of its sub-functions.
Finally, in Section \ref{sec:FAI-LCD} we present a
coding-theory characterization of perfect algebraic immune Boolean functions
by means of the LCD-ness of punctured Reed-Muller codes and derive some new infinite families of LCD codes.

\section{Preliminaries and notation}\label{preliminaries}
In this section, we give a brief introduction to algebraic immunity,
 Reed-Muller codes and linear complementary dual codes, which
are the foundations of other sections.
\subsection{Algebraic immunity of Boolean functions}
Let $n$ be any positive integer. In this paper, we shall denote by
$\B_n$ the set of all $n$-variable Boolean functions over $\F^n$.
Any $n$-variable Boolean function $f$ (that is a mapping from
$\F^n$ to $\F$) admits a unique \emph{algebraic normal form} (ANF),
that is, a representation as a multivariate polynomial over $\F$
\begin{displaymath}
f(x_1,\ldots,x_n)=\sum_{I\subseteq \{1,\ldots ,n\}}
a_I\,\prod_{i\in I}{x_i},
\end{displaymath}
where the $a_I$'s are in $\F$. The terms $\prod_{i\in I} {x_i}$ are
called {\em monomials}.  The {\em algebraic degree} $\deg(f)$ of a
Boolean function $f$ equals the maximum degree of those monomials
whose coefficients are nonzero in its algebraic normal form.

If we identify $\F^n$ with the Galois field $\GF n$  of order
$2^n$, Boolean functions of $n$-variables are then the binary functions over the Galois field $\GF n$ (one can always endow this
vector space with the structure of a field, thanks to the choice of a basis of $\GF n$ over $\F$)
 of order
$2^n$.
The \emph{support} of $f$, denoted by $supp(f)$, is the set of elements of $\GF{n}$ whose image under $f$ is $1$, that is,
$supp(f)=\left \{  x\in \GF{n}: f(x)=1 \right \}$.
The {\em weight} of $f$, denoted by $\wt(f)$, is the {\em
Hamming weight} of the image vector of $f$, that is, the cardinality
of its support $supp(f):=\{x\in \GF{n}\mid f(x)=1\}$.

For any positive integer $k$, and $r$ dividing $k$,
the trace function from $\GF{k}$ to $\GF{r}$, denoted by
$\tr[r]{k}$, is the mapping defined as: $\forall x\in\GF{k} ,\quad
  \tr[r]{k}(x):=\sum_{i=0}^{\frac kr-1}
  x^{2^{ir}}$. In particular, we denote the {\em absolute trace} over $\GF{}$ of an element $x \in
\GF {n}$ by $\tr{n}(x) = \sum_{i=0}^{n-1} x^{2^i}$.
Every non-zero Boolean function $f$ defined on $\GF{n}$ has a (unique)
trace expansion of the form:
\begin{displaymath}
  \forall x\in\GF n ,\quad
  f(x)=\sum_{j\in\Gamma_n}\tr{{o(j)}}(a_jx^{j}) + \epsilon
  (1+x^{2^n-1})
\end{displaymath}
called its polynomial form, where $\Gamma_n$ is the set of integers
obtained by choosing one element in each cyclotomic class of $2$
modulo $2^n-1$, $o(j)$ is the size of the cyclotomic coset of $2$
modulo $2^n-1$ containing $j$, $a_j \in \GF{o(j)}$ and,
$\epsilon=\wt(f)$ modulo $2$.  The algebraic degree of $f$ is equal to
the maximum 2-weight of an exponent $j$ for which $a_j\not= 0$ if
$\epsilon=0$ and to $n$ if $\epsilon=1$. We recall that the $2$-weight
of an exponent $j$, that we denote by $w_2(j)$, is the number of
$1$ in its binary expansion.

An $n$-variable Boolean function $g$ is said to be an {\it
  annihilator} of an $n$-variable Boolean function $f$ if
$f\cdot g=0$, where $f \cdot g$ is the Boolean function whose output equals the product in $\F$ of the outputs of $f$ and $g$. The set of all non-zero annihilators of a Boolean
function $f$ shall be denoted by $\AN(f)$. We shall denote by $\AN^c(f)$
the complement of the set $\AN(f)$, that is, the set of $n$-variable
Boolean function $f$ such that $f\cdot g\not=0$. We shall denote by
$\LDA(f)$ the minimum degree of non-zero annihilators of $f$
\cite{MeierPasalicCarlet2004}. The {\it algebraic immunity} $\AI(f)$
of $f$ is the minimum value between $\LDA(f)$ and $\LDA(1+f)$. Obviously, for  a Boolean function $f$, we have $\AI(1+f)=\AI(f)$. In addition, algebraic immunity is invariant under affine transformations. More specifically, if  $f$ is an $n$-variable Boolean function and $A$ is an affine
  automorphism of $\V n$ then, $\AI(f)=\AI(f\circ A)$.

\subsection{Linear codes and Reed-Muller codes}
An $[\ell,k]$ code $\mathcal C$ over the finite field $\mathbb F_q$ is a linear subspace of
$\mathbb F_q^{\ell}$ with dimension $k$. For convenience, we will denote by $\mathrm{dim}_{\mathbb F_q}(\mathcal C)$
the dimension of the code $\mathcal C$. The dual code, denoted by $\mathcal C^{\perp}$, of $\mathcal C$ is a linear
code with dimension $\ell - k$ and is defined by
$$\mathcal C^{\perp}= \left \{ ( w_{i})_{i=1}^{\ell}\in \mathbb F_q^{\ell}:  \begin{array}{l}
  c_1w_1 + \cdots +c_{\ell } w_{\ell } =0\\
  \text{for all } (c_1, \cdots, c_{\ell})\in \mathcal C
  \end{array}
\right \}.$$

Puncturing and shortening  are classical techniques used to obtain codes of length less than $\ell$ from  mother codes having length $\ell$, thus decreasing the length of codes.
These constructions will be useful for understanding fast algebraic immunity.
Given an $[\ell, k]$ code $\mathcal C$, we can puncture it by deleting the same coordinate $i$ in each codeword. The resulting code is denoted by $\mathcal C^{\{i\}}$.
For any set $D$ of coordinates in $\mathcal C$, we use $\mathcal C^{D}$ to denote the code obtained by puncturing $\mathcal C$ in all coordinates in $D$.
Let $\mathcal C (D)$ be the set of codewords which are $\mathbf 0$ on $D$. Then the shortened code of $\mathcal C$ in all coordinates in $D$, denoted by $\mathcal C_{D}$,
is the code obtained by puncturing the coordinates of $\mathcal C(D)$ in $D$. Puncturing a code is equivalent to shortening the dual code, as explained by the following
proposition,  whose proof can be found in \cite[Theorem 1.5.7]{HP03}.
\begin{proposition}\label{prop:punc-shor-dual}
Let $\mathcal C$ be an $[\ell, k]$ code over $\mathbb F_q$ and let $D$ be any set of coordinates of $\mathcal C$. Then
$$\left ( \mathcal C^D \right ) ^{\perp} = \left ( \mathcal C^{\perp} \right )_{D} \text{ and }
 \left ( \mathcal C_D \right ) ^{\perp} = \left ( \mathcal C^{\perp} \right )^{D}.$$
\end{proposition}

Reed-Muller codes first appeared in print in 1954 and remain ``$\cdots$  one of the oldest and best understood families of codes" \cite[p. 370]{MS86}.
As stated in \cite{Book-Carlet}, we use $\mathrm{RM}(d,n)$ to denote the $d$th-order
Reed-Muller code of length $2^n$. Each codeword in $\mathrm{RM}(d,n)$ is defined by evaluating an $n$-variable Boolean function
$h$ of degree at most $d$ at all points in $\mathbb F_2^n$. The Reed-Muller codes $\mathrm{RM}(d,n)$ have been shown to be equivalent to primitive cyclic codes
(codes of length $2^n-1$) with an overall parity check added \cite{KLP68}. Let $\alpha$ be a primitive element of $\mathbb F_{2^n}$.
Let  $P_0=0$ and $P_j=\alpha^{j-1}$, where
$1\le j \le 2^n-1$.
Then $P_0,...,P_{2^n-1}$ is an enumeration of the points of the vector space $\mathbb F_{2^n}$.
Under this enumeration, the Reed-Muller code $\mathrm{RM}(d, n)$ of order $d $ in $n$ variables may be written as
\begin{eqnarray*}
\mathrm{RM}(d, n)=\left \{(f(P_0), \cdots, f(P_{2^n-1})):  f\in \B _n, \mathrm{deg}(f)\le d \right \}.
\end{eqnarray*}

   We summarize the results on the properties of the Reed-Muller codes in the following theorem.
   For the details of proof we refer the reader to \cite{KLP68,MS86}.
   \begin{theorem}\label{thm:RM}
 Let $n$ be any positive integer and let $0\le d \le n$.
 \begin{enumerate}
 \item[(1)] The Reed-Muller code $\mathrm{RM}(d, n)$ is a binary linear code of dimension $\sum_{i=0}^{d} \binom{n}{i}$.
 \item[(2)] The dual code of $\mathrm{RM}(d, n)$ is the code $\mathrm{RM}(n-d-1, n)$.
 \end{enumerate}

   \end{theorem}

\subsection{Linear complementary dual codes and self-orthogonal codes}
The \emph{hull} of a linear code $\mathcal C$ is defined to be $\mathrm{Hull}(\mathcal C) =\mathcal C \cap \mathcal C^{\perp}$.
When $\mathrm{Hull} (\mathcal C)= \mathcal C$,
$\mathcal C$ is
said to be \emph{self-orthogonal}.
In particular, if $\mathrm{Hull}(\mathcal C) =\mathcal C = \mathcal C^{\perp}$,
then $\mathcal C$ is called a self-dual
code.
 The code $\mathcal C$ is called a \emph{linear complementary dual code} (in brief, an \emph{LCD} code) if $\mathrm{Hull}(\mathcal C)$ is the zero space. These codes have been
 extensively studied recently \cite{CMTQ-DCC,CMTQ-SIG,CMTQP2018,JX,LDL17,LLD17,MTQ}.

 For a matrix $G$, $G^T$ denotes the transposed matrix of $G$. The \emph{Gram matrix} of $G$ is defined to be $GG^{T}$.
 The Gram matrix of a generator of a linear code plays an important role in the study of
 the hulls of linear codes  \cite{GJG18,ZLTD}.
\begin{proposition}\label{prop:LCD-Gram}
 Let $\mathcal C$ be an $[\ell, k]$ linear code over $\mathbb F_q$ with generator matrix $G$. Then
 \begin{eqnarray*}
 \mathrm{dim}_{\mathbb F_q} \left ( \mathrm{Hull} (\mathcal C) \right )=k-\mathrm{Rank} \left (GG^{T}\right).
 \end{eqnarray*}

 In particular,  $\mathcal C$ is LCD (resp. self-orthogonal) if and only if
 $GG^T$ is nonsingular  (resp. $GG^T = 0$).
 \end{proposition}

 A vector $\mathbf x=(x_1, x_2, \ldots, x_{\ell})$ in $\mathbb F_2^{\ell}$ is even-like if
$\sum_{i=1}^{\ell}  x_i=0.$
A binary code is said to be \emph{even-like} if it has only even-like codewords.
 The following proposition gives a necessary condition for an even-like code being LCD
\cite{CMTQ-cha}.
\begin{proposition}\label{prop:dim-even}
Let $\mathcal C$ be an even-like binary code with parameters $[\ell, k]$. If $\mathcal C$ is LCD then  $k$ is an even integer.
\end{proposition}

\section{Fast algebraic immunity and fast immunity profile}

\subsection{A new definition of fast algebraic immunity and its consequences}\label{sec:new-definition-fast}
In the literature, different criteria have been proposed to characterize
the immunity of Boolean functions against fast algebraic attacks; some of those characterizations do not define a parameter but
a property that should satisfy a Boolean function to resist to fast
algebraic attacks \cite{Gong2008,Liu2012,Pasalic2008}.

\begin{definition}\label{dfi:immunite_algebrique}
  Let $f$ be an $n$-variable Boolean function.  We call {\it fast
    algebraic immunity} of $f$, denoted by $\FAI(f)$, the smallest
  value taken by $\deg(g)+\deg(f\cdot g)$ when $g\not \equiv 1$ ranges over the set
  $\AN^c(f)$. We say that such a $g$ {\it achieves} $\FAI(f)$.
\end{definition}

\begin{remark}\label{rem:courtois}
  It has been shown in \cite[Theorem 7.2.1]{Courtois2003} that, for
  any $n$-variable Boolean function $f$ and for every positive
  integers $d$ and $e$ such that $d+e\geq n$, there exists $g$ of
  algebraic degree at most $d$ and $h$ of algebraic degree at most $e$
  such that $f\cdot g=h$.  It implies that $\FAI(f)\leq n$.  Now, one
  has $\deg(f\cdot g)\geq \deg(g)$ if $\deg(f\cdot g)+\deg(g)=\FAI(f)$
  (
 indeed, suppose that $\deg(f \cdot g)< \deg(g)$ then $g$ cannot achieve $\deg(f \cdot g) + \deg(g) = \FAI(f)$, because, denoting $h=f\cdot g$, we have $f\cdot h=h$ and then $h\neq 0$ and $\deg(f \cdot h) + \deg(h)<\deg(f \cdot g) + \deg(g)$). Hence, if $g$ achieves
  $\FAI(f)$ then, necessarily, $\deg(g)\leq\frac n2$ since $\deg(g)+\deg(f \cdot g) =\FAI(f)\le n$ and $\deg(g) \le \deg(f \cdot g)$
  implies that $\deg (g) \leq \frac n2$.
\end{remark}

\begin{remark}\label{rem:courtois}
  In \cite{cryptoeprint:2015:435,Liu2012}, the authors proposed
  different criteria that should satisfy a Boolean function to be
  (almost) resistant to fast algebraic attacks. Those criteria are
  very similar.  Indeed, in \cite[Definition
  1]{cryptoeprint:2015:435}, it is defined that an $n$-variable
  Boolean function $f$ would be \emph{almost optimal resistant}
  against fast algebraic attacks if, for $1\leq e < \frac n2$,
  $\deg(f\cdot g)\geq n-e-1$ whenever $\deg(g)\leq e$ and
  $f\cdot g\not=0$. In \cite[Definition 2]{Liu2012}, the authors
  defined \emph{perfect algebraic immune functions} as the $n$-variable
  Boolean function $f$ such that, for every $1\leq e<\frac n2$,
  $\deg(f\cdot g)\geq n-e$ for any $n$-variable Boolean function $g$
  of algebraic degree at most $e$.

  Observe that the \emph{almost optimal resistance} defined in
  \cite{cryptoeprint:2015:435} is equivalent to $\FAI(f)\geq n-1$
  while the \emph{perfect algebraic immune functions} of \cite{Liu2012} are
  those such that $\FAI(f)\geq n$.

\end{remark}


We first derive from the definition of $\FAI$ an upper bound on the
algebraic degree of Boolean function achieving $\FAI$.

\begin{proposition}\label{prop:degre_min}
  Let $n$ be a positive integer.  Let $f$ be an $n$-variable
  function. Let $g$ an $n$-variable function achieving $\FAI(f)$.
  Then $\deg(g)\leq \left \lfloor\frac{\FAI(f)}2 \right \rfloor$ and
  $\deg(f\cdot g)\geq \left \lceil\frac{\FAI(f)}2 \right \rceil$
\end{proposition}
\begin{proof}
  Let $g\in\AN^c(f)$ achieving $\FAI(f)$, that is,
  $\deg(g)+\deg(f.g)=\FAI(f)$. Necessarily $\deg(f\cdot g)\geq\deg(g)$
  (see Remark \ref{rem:courtois}). Hence $2\deg(g)\leq\FAI(f)$, that
  is, $\deg(g)\leq \left \lfloor\frac{\FAI(f)}2 \right \rfloor$. Thus
  $\deg(f\cdot g)=\FAI(f)-\deg(g)\geq
  \FAI(f)- \left \lfloor\frac{\FAI(f)}2 \right \rfloor= \left \lceil\frac{\FAI(f)}2 \right \rceil$.
  
 \qed
\end{proof}


In \cite{589507}, it has been proposed another definition than ours
for the fast algebraic immunity. Indeed, in \cite{589507}, the authors
give the following definition for the fast algebraic immunity of a
Boolean function:
\begin{equation}\label{AI:589507}
  \begin{split}
    \min\left(2\AI(f),\min_{1\leq\deg(g)<\AI(f)}(\deg(g)+\deg(f\cdot g))\right).
  \end{split}
\end{equation}
\begin{remark}
Using the definition above, an upper bound of fast algebraic immunity of power functions has been established by Mesnager and Cohen \cite{MesnagerCohenAMC2016}. More precisely,  let $f(x)=Tr_1^n(\gamma x^d)$ where $\gamma\in\mathbb{F}_{2 n}$ and $d$ is a
  positive integer. Suppose that
  $\AI(f)\geq\Big\lceil\frac n{\lfloor{\sqrt n}\rfloor}\Big\rceil+1$.
  Then
  $$
  \FAI(f) \leq u\lfloor{\sqrt n}\rfloor+
  2\Big\lceil\frac n{\lfloor{\sqrt n}\rfloor}\Big\rceil - 1
  $$
  where $u$ is the number of runs of $1$ in the binary representation
  of $d$.

\end{remark}
Note that $\AN^c(f)$ contains all $n$-variable Boolean functions $g$
such that $1\leq\deg(g)<\AI(f)$ because any $g$ of algebraic degree
less than $\AI(f)$ cannot be an annihilator of $f$. Hence
\begin{displaymath}
  \min_{1\leq\deg(g)<\AI(f)}\left(\deg(g)+\deg(f\cdot g)\right)\geq\FAI(f).
\end{displaymath}
Furthermore, $\AN^c(f)$ may contain Boolean function of algebraic
degree greater than or equal to $\AI(f)$, that is,
$\{g:\V n\rightarrow\V{}\mid 1\leq\deg(g)<\AI(f)\}$ is strictly
contained in $\AN^c(f)$. However, (\ref{AI:589507}) is less than
$\FAI(f)$ if and only only if $2\AI(f)<\FAI(f)$. We are now going to
show that the $\FAI$ of a Boolean function or its complement is
necessarily less than or equal to (\ref{AI:589507}). To this end, we
first show that the $\FAI$ of a function can be bounded from above and
below with the lowest degree of the non-zero annihilators of its
complement.

\begin{proposition}\label{thm:borne_FAI_grossiere}
  Let $f$ be a non-zero $n$-variable Boolean function. Then
  $$ \LDA(1+f)+1 \leq\FAI(f)\leq 2\LDA(1+f).$$
\end{proposition}

\begin{proof}
  Note that $g\not\in\AN^c(f)$ says only that $f\cdot g\not=0$.
  Hence, $\AN^c(f)$ contains all the non-zero annihilators of $1+f$
  since $f\cdot g=g$ for every $g\in\AN(1+f)$.  Now, if $g$ is an
  non-zero annihilator of $1+f$ then $f \cdot g=g$ proving that
  $\FAI(f)\leq 2\deg(g)$ from which we deduce that
  $\FAI(f)\leq 2\LDA(1+f)$.

  On the other hand, observe that, if $f\cdot g\not=0$, then
  $f\cdot g$ is an non-zero annihilator of $1+f$. Thus
  $\deg(g)+\deg(f\cdot g)\geq  1+\LDA(1+f)$. \qed
\end{proof}

\begin{remark}\label{rem:borne_inf_tendue}
  The lower bound and the upper bound in Proposition \ref{thm:borne_FAI_grossiere}
  are achieved. Indeed, let $f$ be an $n$-variable Boolean function
  whose support strictly contains the support of an affine Boolean
  function $l$. Then $fl=l$ which implies that $\LDA(1+f)=1$ since
  $(1+f)\cdot l=0$. Thus
  $\deg(l)+\deg(f\cdot l)= 2\deg(l)=2=\LDA(1+f)+1=2\LDA(1+f)$.
\end{remark}

Proposition \ref{thm:borne_FAI_grossiere} says that, for any
$n$-variable Boolean function, $\FAI(f)\leq 2\LDA(1+f)$. Thus, if
$\LDA(1+f)>\LDA(f)$, $\FAI(1+f)\leq 2\LDA(f)=2\AI(f)$ while
$\FAI(f)\leq 2\AI(f)=2\LDA(1+f)$ if $\LDA(1+f)\leq\LDA(f)$.
Summarizing:
\begin{corollary}\label{cor:FAI<=2AI}
  Let $f$ be an $n$-variable Boolean function. Then
  \begin{displaymath}
    \min(\FAI(f),\FAI(1+f))\leq 2\AI(f).
  \end{displaymath}
\end{corollary}
Based on this observation, an extension of fast algebraic immunity is
given by Definition \ref{dfi:immunite_algebrique1}.
\begin{definition}\label{dfi:immunite_algebrique1}
  Let $f$ be an $n$-variable Boolean function. The $\FAA$ of $f$ is
  the minimum value between $\FAI(f)$ and $\FAI(1+f)$:
  \begin{displaymath}
    \FAA(f) = \min(\FAI(f),\FAI(1+f)).
  \end{displaymath}
\end{definition}

A direct consequence of Proposition \ref{thm:borne_FAI_grossiere} is
then that the $\FAA$ of an $n$-variable function is less than or equal
to (\ref{AI:589507}). But above, one deduces from Proposition
\ref{thm:borne_FAI_grossiere}.

\begin{proposition}
  Let $f$ be an $n$-variable Boolean function. Then
  \begin{displaymath}
   \min \left (  \LDA(f)+1, \LDA(1+f)+1  \right ) \leq\FAA(f)\leq2\AI(f).
  \end{displaymath}
\end{proposition}


Another property of $\FAI$ is that it is invariant under affine
transformations like the standard algebraic immunity.
\begin{proposition}\label{prop:affine_invariance}
  Let $f$ be an $n$-variable Boolean function and $A$ an automorphism
  of $\V n$. Then $\FAI(f\circ A)=\FAI(f)$.
\end{proposition}
\begin{proof}
  Note that
  $\{\deg(g)+\deg(f\cdot g)\mid g\in\AN^c(f)\}=\{\deg(g\circ
  A)+\deg(f\circ A\cdot g\circ A)\mid g\in\AN^c(f)\}$.
  Observe now, that $f\cdot g\not=0$ if and only if
  $f\circ A\cdot g\circ A \not=0$, that is, $g\in\AN^c(f)$ if and only
  if $g\circ A\not\in\AN^c(f\circ A)$. Thus $\FAI(f)=\FAI(f\circ A)$. \qed
\end{proof}

One can extend Proposition \ref{prop:affine_invariance} to $\FAA$.
\begin{proposition}
  Let $f$ be an $n$-variable Boolean function and $A$ an automorphism
  of $\V n$. Then $\FAA(f\circ A)=\FAA(f)$.
\end{proposition}

\subsection{Fast immunity profile}
\label{sec:fast-immun-prof}

Set $\MUL_k(f)=\{f\cdot g\mid \deg(g)\leq k\}$ and
$\mu_k(f)=\mindeg\MUL_k(f)$, where $\mindeg$ denotes the minimum degree
of the non-zeros elements of the set. Clearly
$(\mu_k(f)_{1\leq k\leq n}$ is a non-increasing sequence of integers. We
shall call it the \emph{fast immunity profile} of $f$.  Note
that
$\MUL_k(f\circ A)=\{f\circ A\cdot g\mid \deg(g\circ A)\leq k\}
=\{f\cdot h\mid \deg(h)\leq k\}\circ A=\MUL_k(f)\circ A$
for every affine automorphism of $\V n$, proving that
\begin{lemma}\label{profil_invariant}
  Let $f$ be an $n$-variable Boolean function and $A$ an automorphism
  of $\V n$. Then $\mu_k(f)=\mu_k(f\circ A)$ for every
  $1\leq k\leq n$.
\end{lemma}

We now show that the algebraic immunity and the fast
algebraic immunity  of $f$ can be expressed by means of the
immunity profile. We first recall the relationship between the annihilators of a function $f$ and the multiples of $f+1$.

\begin{proposition}\label{LDAMUL}
  Let $f$ be an $n$-variable Boolean function. Then

  $$\LDA(f) = \min_{1\leq k\leq n}\mu_k(f+1).$$

  Furthermore, if $k\geq LDA(f)$, $\mu_k(f+1)=\LDA(f)$.
\end{proposition}
\begin{proof}
  For any integer $k$ ranging from $1$ to $n$, we have
  $\mu_k(1+f)=\mindeg\MUL_k(1+f)\geq\LDA(f)$ since every nonzero
  element of $\MUL_k(1+f)$ is a non-zero annihilator of $f$. It follows that
   $\min_{1\leq k\leq n}\mu_k(f+1)\geq \LDA(f)$.

  Conversely, let $g$ be an annihilator of $f$ of algebraic degree
  $\LDA(f)$. Then one has $(1+f)\cdot g=g$ and thus
  $\mu_{\LDA(f)}(1+f)\leq\LDA(f)$ implying that
  $\min_{1\leq k\leq n}\mu_k(f+1)\leq \LDA(f)$.
  Consequently, $\LDA(f) = \min_{1\leq k\leq n}\mu_k(f+1).$

  Furthermore, note that $(\mu_k(1+f))_{1\leq k\leq n}$ is a
  nonincreasing sequence of integers. Hence, since
  $\mu_{\LDA(f)}(1+f)\leq\LDA(f)$, one has necessarily
  $\mu_{k}(1+f)=\LDA(f)$ when $k\geq\LDA(f)$. \qed
\end{proof}

Recalling that
$\AI(f)=\min(\LDA(f),\LDA(1+f))$, we deduce:

\begin{proposition}\label{AIMUL}
  Let $f$ be an $n$-variable Boolean function. Then
  $$\AI(f)=\min(\min_{1\leq k\leq n}\mu_k(f+1),\min_{1\leq k\leq n}\mu_k(f))$$

\end{proposition}

\begin{proposition}\label{FAIMUL}
  Let $f$ be an $n$-variable Boolean function. Then
  $$\FAI(f)=\min_{1 \leq k\leq n}(k+\mu_k(f)).$$
\end{proposition}
\begin{proof}
  Let $1\leq k\leq n$ be arbitrary.  By definition, $\mu_k(f)$ is the
  lowest algebraic degree of all nonzero elements of
  $\MUL_k(f)$. Thus, for $g \not =0$, $\deg(g)=k$, $f\cdot g\not=0$,
  one has $\deg(f\cdot g)+\deg(g)\geq\mu_k(f)+k$.  Hence, one gets
  that $\FAI(f)\geq\min_{1\leq k\leq n}(k+\mu_k(f))$.  Conversely, let
  $1\leq j\leq n$ be such that
  $j+\mu_j(f)=\min_{1\leq k\leq n}(k+\mu_k(f))$.  Let $h$ be a
  function with $\deg(h)=j$ achieving $\mu_j(f)$ (such that
  $\mu_j(f)=\deg(f\cdot h)$ and $f\cdot h\not=0$).  Then
  $\deg(h)+\deg(f\cdot h)=j+\mu_j(f)\geq \FAI(f)$. \qed
\end{proof}

\section{Fast algebraic immunity, approximation and concatenation of functions}
\subsection{Fast algebraic immunity and approximation of functions}
\label{sec:fast-immun-algebr}

In \cite{Zhang2006} the algebraic complement of a Boolean function and  its algebraic immunity have been studied.

\begin{definition}
  Given a Boolean function $f$ defined on $\V n$, the algebraic
  complement of $f$, denoted by $f^c$, is the function that contains
  all the monomials that are not in the algebraic normal form of $f$.
\end{definition}

In \cite[Theorem 2]{Zhang2006}, the authors have shown that the
algebraic immunities of a Boolean function and its algebraic
complement are close:
\begin{displaymath}
  \vert AI(f) - \AI(f^c) \vert \leq 1.
\end{displaymath}

Let us denote by $\delta_0$ the indicator of the singleton $\{0\}$. It is well-known and easily checked that the algebraic normal form of $\delta_0$ equals $\sum_{I\subseteq \{1,\ldots ,n\}}
\prod_{i\in I}{x_i}$. The algebraic complement of a function $f$ is then the function $f+\delta_0$. Since the algebraic immunity is invariant under affine transformations, there is no reason to privilege $\delta_0$ rather than any other indicator of a singleton $\delta_a$ (except that the ANF of the algebraic complement is nicely simple). Moreover, functions $f+ \delta_a$, $a\in \V n$, are all functions at Hamming distance 1 from $f$, and it seems natural to consider more generally functions at low Hamming distance from $f$. A nice observation has been made in \cite{JohanssonWang2010}: if a function is at low Hamming distance from a low algebraic degree function, then it is weak against fast algebraic attacks. We show now that if a function is at low Hamming distance from a low algebraic immunity function, then it is weak against (standard) algebraic attacks:

\begin{proposition}let $k$ and $d$ be two positive integers. Let $f$ be any $n$-variable Boolean function such that $AI(f)=k$. Let $\delta$ be any Boolean function such that $w_H(\delta)< \min(2^{n-k},2^{d+1}-1)$. Then: $$|AI(f+\delta)-AI(f)|\leq d.$$
\end{proposition}
\begin{proof}
There exists by hypothesis a nonzero annihilator $g$ of $f$ or of $f+1$ whose algebraic degree is $k$.
Let $a_0$ be any element such that $\delta(a_0)=0$ and $g(a_0)=1$.  Such an element exists because the Hamming weight of $g$ is larger than or equal to $2^{n-k}$.

Write $\mathrm{Supp}(\delta)= \left \{ x \in \mathbb F_2^n: \delta (x)=1 \right \}$ and $w=w_H(\delta)$. Let the points
of $\mathrm{Supp}(\delta)$ be $a_1, \cdots, a_w $.  Let $\pi$ be the linear mapping given by
\begin{eqnarray*}
\begin{array}{rl}
\pi: \mathrm{RM}(d,n) & \longrightarrow \mathbb F_2^{w+1}\\
(h(x))_{x\in \mathbb F_2^n} & \mapsto (h(a_i))_{i=0}^{w}.
\end{array}
\end{eqnarray*}
We next claim that the mapping $\pi$ is surjective.
Suppose the claim was false. Then the image of $\pi$ lies  in a hyperplane of $\mathbb F_2^{w+1}$ and
thus we could find a non-zero vector $(c_{a_i})_{i=0}^{w} \in \mathbb F_2^{w+1}$
such that $c_{a_0} h(a_0) + \cdots c_{a_w} h(a_w)=0$ for any $n$-variable Boolean function $h$ of
degree at most $d$. It follows that the vector $(c_x)_{x\in \mathbb F_2^n}$
defined by $c_x= \left \{ \begin{array}{rr}
 c_{a_i}, & \text{ if } x=a_i\\
 \\
 0, & \text{ otherwise}
 \end{array} \right .
$ belongs to the dual code $\mathrm{RM}(d,n)^{\perp}= \mathrm{RM}(n-d-1,n) $ of $\mathrm{RM}(d,n)$.
Note that the weight of the codeword $(c_x)_{x\in \mathbb F_2^n}$ of $\mathrm{RM}(n-d-1,n) $
is less or equal to $w+1$, which contradicts the facts that the minimum distance of
$\mathrm{RM}(n-d-1,n) $ is at least $2^{d+1}$ and $w+1\le 2^{d+1}-1$. Thus $\pi$ is a surjective mapping.
In particular,  there exists a polynomial $h$ of
degree at most $d$ such that $h(a_0)=1$ and $h(a_i)=0$ for $1\le i \le w$.
We have then $(f+\delta)\cdot gh=0$ or $(f+1+\delta)\cdot gh=0$ and $gh$ is an annihilator of $f+\delta$ or of $f+1+\delta$ and
it is nonzero since $(gh)(a_0)=1$. This implies that $AI(f+\delta)\leq AI(f)+d$ and applying this result to $f+\delta$ instead of $f$ gives $AI(f)\leq AI(f+\delta)+d$, which completes the proof. \qed
\end{proof}

Note that this result and the result from \cite{JohanssonWang2010} mentioned above are complementary of each other since the condition of being at low Hamming distance from a low algebraic immunity function is a weaker assumption than being at low Hamming distance from a function of low algebraic degree, and moreover the weakness against standard algebraic attacks is still worse than the weakness against fast algebraic attacks (because when they apply, algebraic attacks are more efficient than fast algebraic attacks), but the result from \cite{JohanssonWang2010} still applies for functions at low Hamming distance from a function whose algebraic degree is not necessarily low, but is not high either; indeed it says that if $w_H(\delta)<\sum_{i=0}^d {n\choose i}$ and $f$ has algebraic degree $k$ then $FAI(f+\delta)\leq k+2d$.\\

Let us now investigate if $\FAI(f)$ and $\FAI(f^c)$ are close or
not. To this end, we shall need the following Lemma.
\begin{lemma}\label{lem:util_encadrement_FAI_alg_compl}
  Let $f\not=\delta_0$ be an $n$-variable Boolean function. Let $g$
  achieving $\FAI(f)$. Then there exists an $n$-variable affine
  function $l$ vanishing at $0$ such that $f\cdot g\cdot l\not=0$.
\end{lemma}
\begin{proof}
  Suppose that for every $n$-variable affine Boolean function $l$
  vanishing at $0$, $f\cdot g\cdot l=0$. Then, for any
  $1\leq i\leq n$, $f\cdot g\cdot l_i=0$ where
  $l_i(x_1,\dots,x_n)=x_i$, that is, $f\cdot g=(1+l_i)h_i$ for some
  $n$-variable Boolean function $h_i$. Therefore,
  $f\cdot g=\delta_0=\prod_{i=1}^nl_i$.  Now,
  $FAI(f)=\deg(f\cdot g)+\deg(g)=n+\deg(g)\leq n$. Hence, $g=1$
  contradicting $f\not=\delta_0$. \qed
\end{proof}

We begin with showing.

\begin{proposition}
  Let $f$ be an $n$-variable Boolean function. Suppose $f\not=\delta_0$
  and $f^c\not=\delta_0$. Then
  $$\vert\FAI(f^c)-\FAI(f)\vert\leq 2.$$
\end{proposition}
\begin{proof}
  According to Proposition \ref{FAIMUL},
  $\FAI(f)=\min_{1\leq k\leq n}(k+\mu_k(f))$.  Let $k\geq 1$
  achieving $\FAI(f)$: $k+\mu_k(f)=\FAI(f)$.
   Let $g$ achieving $\mu_k(f)$:
  $\deg(f\cdot g)=\mu_k(f)$.

  Now, observe that, for any $n$-variable Boolean function $p$,
  $$
  f^c\cdot p
  =f\cdot p + \delta_0\cdot p =
  \begin{cases}
    f\cdot p & \mbox{if $p(0)=0$}\\
    f\cdot p + \delta_0 & \mbox{ if $p(0)=1$}
  \end{cases}
  $$
  But in all cases, for any $n$-variable affine Boolean function $l$
  vanishing at $0$,  $f\cdot g\cdot l=f^c\cdot g\cdot l$ since
  $\delta_0\cdot l=0$. According to Lemma
  \ref{lem:util_encadrement_FAI_alg_compl}, there exists $l$ such that
  $f\cdot g\cdot l=f^c\cdot g\cdot l\not=0$. Then
  \begin{displaymath}
    \begin{split}
      \FAI(f^c) &\leq\deg(f^c\cdot g\cdot l)+\deg(g\cdot l)\\
      &\leq\deg(f^c\cdot g)+\deg(g) + 2 = \FAI(f)+2.
    \end{split}
  \end{displaymath}
  Now, since the algebraic complement of $f^c$ is $f$ itself.
  One can exchange the role of $f$ and its algebraic complement $f^c$ in the
  above arguments and prove
  \begin{displaymath}
    \FAI(f)\leq\FAI(f^c)+2.
  \end{displaymath}
  \qed
\end{proof}
\begin{remark}
  Following the above proof, if the $n$-Boolean function $g$ achieving
  $\FAI(f)$ vansihes at $0$. Then, one has
  $\FAI(f^c)\leq\deg(f^c\cdot g)+\deg(g)=\deg(f\cdot
  g)+\deg(g)=\FAI(f)$.
  Therefore, if $\FAI(f^c)$ is also achieved by an $n$-variable
  Boolean function vanishing at $0$ then,
  $\FAI(f^c)\geq\FAI(f)$. Therefore, we might have $\FAI(f)=\FAI(f^c)$
  for some subclasses of $n$-variable Boolean functions.
\end{remark}
\begin{remark}
  Observe that the condition $f\not=\delta_0$ is not restrictive since
  $\AI(\delta_0)=1$ ($\delta_0\cdot l=0$ if $l(0)=0$ and $\deg(l)=1$).
\end{remark}

\subsection{Fast algebraic immunity and concatenation of Boolean functions}
\label{sec:fast-algebr-immun}

A classical secondary constructions of Boolean functions from Boolean functions in lower dimension
is the following.
\begin{definition}
  Let $f_0$ and  $f_1$ be two $(n-1)$-variable  Boolean functions. The
  concatenation  of  $f_0$  with  $f_1$  is  the  $n$-variable  Boolean
  function defined, for $x=(x_1,\dots,x_n)\in\V n$, by
  \begin{equation}\label{eq:concatenation}
    \begin{split}
    f(x_1,\dots,x_n) &= (x_n+1)f_0(x_1,\dots,x_{n-1})\\
    &\qquad+ x_nf_1(x_1,\dots,x_{n-1})
    \\&=
    \begin{cases}
      f_0(x_1,\dots,x_{n-1}) &\mbox{if $x_n=0$}\\
      f_1(x_1,\dots,x_{n-1}) &\mbox{if $x_n=1$}.
    \end{cases}
  \end{split}
\end{equation}
\end{definition}

Any $n$-variable Boolean function of algebraic degree $k$ can be
written
\begin{equation}\label{eq:g}
  \begin{split}
    g(x_1,\dots,x_n) &=  (x_n+1) g_0(x_1,\dots,x_{n-1}) \\
    &\quad + x_n g_1(x_1,\dots,x_{n-1})
  \end{split}
\end{equation}
where $g_0$ and $g_1$ are $(n-1)$-Boolean function and
\begin{displaymath}
  \deg(g) = \max(\deg(g_0),\deg(g_0+g_1)+1).
\end{displaymath}
Observe that, the product $f\cdot g$, where $f$ is given by
(\ref{eq:concatenation}) and $g$ is given by (\ref{eq:g}), is
\begin{displaymath}
  \begin{split}
    &f(x_1,\dots,x_n)g(x_1,\dots,x_n) \\
    &\quad=  (x_n+1)f_0(x_1,\dots,x_{n-1})g_0(x_1,\dots,x_{n-1}) \\
    &\qquad + x_nf_1(x_1,\dots,x_{n-1})g_1(x_1,\dots,x_{n-1}).
  \end{split}
\end{displaymath}
Hence
\begin{equation}\label{eq:degfg}
  \begin{split}
    &\deg(f\cdot g)+\deg(g)=\max\Big(\deg(f_0g_0)+\deg(g_0),\\
    &\qquad \deg(f_0g_0)+\deg(g_0+g_1)+1,\\
    &\qquad \deg(f_0g_0+f_1g_1)+\deg(g_0)+1,\\
    &\qquad\deg(f_0g_0+f_1g_1)+\deg(g_0+g_1)+2\Big).
  \end{split}
\end{equation}
Based on this observation, we prove
\begin{proposition}\label{prop:concatenation}
  Let $n$ be a positive integer greater than $1$. Let $f_0$ and $f_1$
  be two $(n-1)$-variable Boolean functions. Let $f$ be the
  $n$-variable Boolean function obtained by concatenating $f_0$ with
  $f_1$.  Then
  \begin{displaymath}
    \FAI(f)\geq\min(\FAI(f_0),\FAI(f_1)+1),
  \end{displaymath}
  and
  \begin{displaymath}
    \FAI(f)\leq\min(\FAI(f_0),\FAI(f_1))+2.
  \end{displaymath}
\end{proposition}
\begin{proof}
   Let $g$ an $n$-variable Boolean
  function achieving $\FAI(f)$: $\FAI(f)=\deg(f\cdot g)+\deg(g)$ and
  $f\cdot g\not=0$.  This Boolean function can be written as
  (\ref{eq:g}).  Since $f\cdot g\not=0$, either $f_0\cdot g_0\not=0$
  either $f_1\cdot g_1\not=0$. If $f_0\cdot g_0\not=0$ then, according
  to (\ref{eq:degfg}),
  $\FAI(f)\geq\deg(f_0g_0)+\deg(g_0)\geq\FAI(f_0)$. If $f_0g_0=0$ and
  $f_1g_1\not=0$, then (\ref{eq:degfg}) rewrites as
  \begin{displaymath}
    \begin{split}
      &\deg(f\cdot g)+\deg(g)=\max\Big(
      \deg(f_1g_1)+\deg(g_0)+1,\\
      &\qquad\deg(f_1g_1)+\deg(g_0+g_1)+2\Big).
    \end{split}
  \end{displaymath}
  The result follows then from nothing that
  \begin{itemize}
  \item either $\deg(g_0+g_1)<\deg(g_1)$, wich implies
    $\deg(g_0)=\deg(g_1)$ and thus
    $\FAI(f)\geq \deg(f_1g_1)+\deg(g_0)+1\geq\FAI(f_1)+1$,
  \item either $\deg(g_0+g_1)\geq\deg(g_1)$ which implies that
     $\FAI(f)\geq \deg(f_1g_1)+\deg(g_1)+2\geq\FAI(f_1)+2\geq\FAI(f)+1$.
  \end{itemize}
  Conversely, if we take $g_0$ achieving $\FAI(f_0)$ and $g_1=0$ in
  (\ref{eq:degfg}), then, we get
  \begin{displaymath}
    \FAI(f)\leq \deg(g) = \deg(f_0g_0)+\deg(g_0)+2=\FAI(f_0)+2.
  \end{displaymath}
  Likewise, if we take $g_1$ achieving $\FAI(f_1)$ and $f_0=0$ then,
  \begin{displaymath}
    \FAI(f) \leq \deg(g) = \deg(f_1g_1)+\deg(g_1)+2=\FAI(f_1)+2.
  \end{displaymath}
  \qed
\end{proof}

In \cite{CarletTang2015}, the authors have considered such
construction to design Boolean functions suitable for the filter model
of pseudo-random generator. More precisely, they have considered the particular
case of the concatenation of an $(n-1)$-variable Boolean function $f$
with its complement $1+f$ to $1$. Let us denote $\bar f$ such a concatenation:
\begin{equation}\label{eq:fbar}
  \begin{split}
    \bar f(x_1,\dots,x_n) &= x_n + f(x_1,\cdots,x_{n-1})\\
    &= (x_n+1)f(x_1,\cdots,x_{n-1})\\
    &\quad + x_n (1+f(x_1,\cdots,x_{n-1})).
  \end{split}
\end{equation}
We then deduce from Proposition \ref{prop:concatenation}
\begin{corollary}
  Let $f$ be an $(n-1)$-variable Boolean function. Let $\bar f$ be
  defined by (\ref{eq:fbar}). Then
  \begin{displaymath}
    \min(\FAI(f),\FAI(1+f)+1)\leq\FAI(\bar f)\leq\FAA(f)+2.
  \end{displaymath}
\end{corollary}
Now, note that $1+\bar f$ is the concatenation of $1+f$ with $f$:
\begin{displaymath}
  \begin{split}
    1 + \bar f(x_1,\dots,x_n) &= 1 + x_n + f(x_1,\cdots,x_{n-1})\\
    &= (x_n+1)(1+f(x_1,\cdots,x_{n-1}))\\
    &\quad + x_n f(x_1,\cdots,x_{n-1}).
  \end{split}
\end{displaymath}
Therefore
\begin{corollary}
  Let $f$ be an $(n-1)$-variable Boolean function. Let $\bar f$ be
  defined by (\ref{eq:fbar}). Then
  \begin{displaymath}
    \FAA(f)\leq\FAA(\bar f)\leq\FAA(f)+2.
  \end{displaymath}
\end{corollary}

\section{Fast algebraic immunity and LCD codes}\label{sec:FAI-LCD}
In this section we shall establish the relation between fast algebraic immunity, perfect algebraic immune functions, punctured Reed-Muller codes and binary LCD codes.

The link between the algebraic immunity of Boolean functions and the dimensions of punctured Reed-Muller codes is
described in the following.
\begin{proposition}\label{prop:AI-dim}
 Let $e$ be a positive integer.  Let $f$ be an $n$-variable Boolean function and let $D$ be its support.
 Then the algebraic immunity of $f$ is greater than $e$ if and only if the dimensions of the two punctured Reed-Muller codes
 $\mathrm{RM}(e,n)^{D}$ and $\mathrm{RM}(e,n)^{\overline{D}}$ are both equal to $\mathrm{dim}_{\mathbb F_2} \left ( \mathrm{RM(e,n)} \right )$.
\end{proposition}
\begin{proof}
Let $\AI(f)>e$. Assume by way of contradiction,
\begin{eqnarray}\label{eq:dim<dim}
\mathrm{dim}_{\mathbb F_2} \left (\mathrm{RM}(e,n)^{\overline{D}} \right )< \mathrm{dim}_{\mathbb F_2} \left ( \mathrm{RM(e,n)} \right ).
\end{eqnarray}
Consider the linear transformation $Res_{D}$ from $\mathrm{RM(e,n)}$ to the punctured code $\mathrm{RM}(e,n)^{\overline{D}}$ defined by
\begin{eqnarray*}
(g(x))_{\mathbb{F}_2^n} \mapsto (g(x))_{x\in D}.
\end{eqnarray*}
By assumption in (\ref{eq:dim<dim}), there exists a nonzero function $g\in \B_n$ of degree at most $e$ such
that $Res_D(g)=(g(x))_{x\in D}=\mathbf{0}$. Then $fg=0$, contrary to $\AI(f)>e$. Hence
$\mathrm{dim}_{\mathbb F_2} \left (\mathrm{RM}(e,n)^{\overline{D}} \right )= \mathrm{dim}_{\mathbb F_2} \left ( \mathrm{RM(e,n)} \right ).$
By a similar argument, we can show
 that $\mathrm{dim}_{\mathbb F_2} \left (\mathrm{RM}(e,n)^{{D}} \right )= \mathrm{dim}_{\mathbb F_2} \left ( \mathrm{RM(e,n)} \right ).$

As for the converse, suppose the assertion is false. Then we could find a nonzero function $g$ of degree at most $e$
such that $fg=0$ or $(1+f)g=0$. By symmetry, one can assume that $fg=0$. We then have $Res_D(g)=(g(x))_{x\in D}$
is the all zeros codeword of $\mathrm{RM}(e,n)^{\overline{D}}$. We conclude that
the linear transformation $Res_D$  is surjective but not injective. This clearly forces
$\mathrm{dim}_{\mathbb F_2} \left (\mathrm{RM}(e,n)^{\overline{D}} \right )< \mathrm{dim}_{\mathbb F_2} \left ( \mathrm{RM(e,n)} \right ),$
a contradiction. Therefore $\AI (f) >e$. \qed
\end{proof}

Given a nonzero Boolean function $f$, let $\left < \mathbf{1}_{\wt (f)} \right >$
denote the binary code generated by the all-ones vector $\mathbf{1}_{\wt (f)}$ of length $\wt (f)$.
To treat fast algebraic immunity of Boolean functions, we need to invoke punctured Reed-Muller codes.
\begin{lemma}\label{lem:deg-dual}
 Let $e, e'$ and $n$ be positive    integers. Let $f$ be an $n$-variable nonzero Boolean function and let $D$ be its support.
 Then the intersection of $\mathrm{RM}(e,n)^{\overline{D}}$ and
 $\left ( \mathrm{RM}(e',n)^{\overline{D}} \right )^{\perp}$
is included in $\left < \mathbf{1}_{\wt (f)} \right >$  if and only if   $\deg (fg) \ge n-e'$ holds
for any $n$-variable nonzero  Boolean function $g\in AN^{c}(f) \setminus \left (1+\AN (f) \right )$
 of degree at most $e$, where $1+\AN (f)=\left \{1+h: h \in \AN (f) \right \}$.
\end{lemma}

\begin{proof}
By Proposition \ref{prop:punc-shor-dual},
\begin{eqnarray*}
\left ( \mathrm{RM}(e',n)^{\overline{D}}  \right )^{\perp} = \left ( \mathrm{RM}(e',n)^{\perp}  \right )_{\overline{D}}.
\end{eqnarray*}
Invoking Part (5) of Theorem \ref{thm:RM}, we get
\begin{eqnarray}\label{eqn:deg-l-h}
\left ( \mathrm{RM}(e',n)^{\overline{D}}  \right )^{\perp} =\mathrm{RM}(n-e'-1,n)  _{\overline{D}}.
\end{eqnarray}

Let us first prove the only if part, so let us suppose that $\mathrm{RM}(e,n)^{\overline{D}}
 \cap \left ( \mathrm{RM}(e',n)^{\overline{D}} \right )^{\perp} \subseteq \left < \mathbf{1}_{\wt (f)} \right > $.
 If there existed a function  $g\in AN^{c}(f) \setminus \left ( 1 +\AN (f)\right )$ such that  $\deg(g)\le e$ and $\deg (fg) \le  n-e'-1$,
 we would have
 \begin{eqnarray*}
 (g(x))_{x \in D} \in \mathrm{RM}(e,n)^{\overline{D}}
 \end{eqnarray*}
 and
 \begin{eqnarray*}
 (g(x)f(x))_{x\in D} \in \mathrm{RM}(n-e'-1,n)_{\overline{D}}.
 \end{eqnarray*}
Then, taking $g\in AN^{c}(f)$ and $supp(f)=D$  into account, one
sees that the nonzero codeword $(g(x))_{x \in D}$ is not equal to $\mathbf{1}_{\wt (f)}$ and lies in the intersection of the punctured code
$\mathrm{RM}(e,n)^{\overline{D}}$ and the shortened code $\mathrm{RM}(n-e'-1,n)_{\overline{D}}$.
From (\ref{eqn:deg-l-h}), we deduce that
$\mathrm{RM}(e,n)^{\overline{D}} \cap \left ( \mathrm{RM}(e',n)^{\overline{D}}  \right )^{\perp}$
is not included in $\left < \mathbf{1}_{\wt (f)} \right >$, a contradiction.
Hence the proof of the only if part is concluded.

For the converse, suppose the assertion of the lemma is false. Then
we could find a function $g \in AN^{c}(f) $ of degree at most $e$ such that
\begin{eqnarray*}
(g(x))_{x\in D} \in \left ( \mathrm{RM}(e,n)^{\overline{D}} \cap \left ( \mathrm{RM}(e',n)^{\overline{D}}  \right )^{\perp} \right ) \setminus \left < \mathbf{1}_{\wt (f)} \right >.
\end{eqnarray*}
Thus, we can apply (\ref{eqn:deg-l-h}) to conclude that $1+g \not \in \AN (f)$ and
$$(g(x)f(x))_{x\in D}=(g(x))_{x\in D}\in \mathrm{RM}(n-e'-1,n)_{\overline{D}}.$$
We thus get $\deg (gf) \le n-e'-1$, a contradiction. This completes the proof. \qed
\end{proof}

\begin{lemma}\label{lem:1-deg-high}
Let $f$ be an $n$-variable nonzero Boolean function and let $D$ be its support. Then
$\mathbf 1_{\wt (f)} \not \in \left ( \mathrm{RM} (e',n)^{\overline{D}} \right )^{\perp}$
if and only if $\deg (f) \ge n-e'$.
\end{lemma}
\begin{proof}
By (\ref{eqn:deg-l-h}), $\mathbf 1_{\wt (f)} \not \in \left ( \mathrm{RM} (e',n)^{\overline{D}} \right )^{\perp}$
if and only if $\mathbf 1_{\wt (f)} \not \in \mathrm{RM} (n-e'-1,n)_{\overline{D}} $.
The desired conclusion then follows from the definition of $\mathrm{RM} (n-e'-1,n)_{\overline{D}}$. \qed
\end{proof}

 The following theorem provides a characterization of fast algebraic immunity of $n$-variable higher degree Boolean functions
 by means of punctured Reed-Muller codes.
\begin{theorem}\label{thm:FAI-s-1-RM}
Let $s$ be a positive integer.  Let $f$ be an $n$-variable nonzero Boolean function with $\deg (f) \ge s-1$
 and let $D$ be its support.
Then the fast algebraic immunity of $f$ is greater than or equal to $s$ if and only if
 $\mathrm{RM}(e,n)^{\overline{D}} \cap \left ( \mathrm{RM}(e+n-s,n)^{\overline{D}} \right )^{\perp} = \left \{  \mathbf{0} \right \}$
  holds for any $1\le e \le n$.
\end{theorem}
\begin{proof}
Let $f$ be an $n$-variable Boolean function with $\FAI (f) \ge s$.
By the definition of fast algebraic immunity, we have
\begin{eqnarray*}
\deg(g) + \deg(gf) \ge s \text{ for any } g \in \AN^c(f) \setminus \{1\}.
\end{eqnarray*}
Now, as then, we can assert that $\deg (gf) \ge s-\deg (g) \ge n-(e +n-s)$
for any $g\in \AN^c(f)$ with $1\le \deg(g) \le e$.
Therefore $\mathrm{RM}(e,n)^{\overline{D}} \cap \left ( \mathrm{RM}(e+n-s,n)^{\overline{D}} \right )^{\perp}  \subseteq \left  < \mathbf{1}_{\wt (f)} \right >$
holds for any $1 \le e \le n$
by Lemma \ref{lem:deg-dual}.
As $\deg (f) \ge s-1 \ge n- (e+n-s)$ we have $\mathrm{RM}(e,n)^{\overline{D}} \cap \left ( \mathrm{RM}(e+n-s,n)^{\overline{D}} \right )^{\perp} =\{\mathbf{0}\} $
from Lemma \ref{lem:1-deg-high}.

Conversely, assume that for any $1 \le e \le n$ one has
$$\mathrm{RM}(e,n)^{\overline{D}} \cap \left ( \mathrm{RM}(e+n-s,n)^{\overline{D}} \right )^{\perp} =\left \{ \mathbf  0\right \}.$$
Suppose the theorem were false. Then we could find a Boolean function $g \in \AN^c (f) \setminus \{1\}$ such that $\deg(g) + \deg(gf) <s$.
It follows that $\deg(gf) <n-(e+n-s) \le \deg (f)$ and $g(1+f) \not \equiv 0$, where $e=\deg(g)$.
Lemma \ref{lem:deg-dual} now implies that $\mathrm{RM}(e,n)^{\overline{D}} \cap \left ( \mathrm{RM}(e+n-s,n)^{\overline{D}} \right )^{\perp}$
is included in $\left < \mathbf{1}_{\wt (f)} \right >$, a contradiction. This completes the proof. \qed
\end{proof}

The following theorem gives a characterization of perfect algebraic immune
functions using the LCD-ness of the punctured codes of Reed-Muller codes by deleting the coordinates outside the
supports of the Boolean functions.
\begin{theorem}\label{thm:PAI-LCD}
Let $f$ be an $n$-variable nonzero Boolean function and let $D$ be its support.
Then $f$ is a perfect algebraic immune function if and only if
 $\mathrm{RM}(e,n)^{\overline{D}}$ is an LCD code for any $1\le e \le n$.
\end{theorem}
\begin{proof}
Let $f$ be a perfect algebraic immune function.
The perfect algebraic immune function $f$  has degree at least $n-1$ (see \cite{Pasalic2008}).
By Theorem \ref{thm:FAI-s-1-RM} and $\FAI (f) = n$,  $\mathrm{RM}(e,n)^{\overline{D}}$ is  LCD  for any $1\le e \le n$.

Conversely, suppose that $\mathrm{RM}(e,n)^{\overline{D}}$ is an LCD code for any $1\le e \le n$.
The desired conclusion then follows from Theorem \ref{thm:FAI-s-1-RM} and Lemma \ref{lem:1-deg-high}. This completes the proof. \qed
\end{proof}

The following corollary has been proved in \cite{Liu2012} and we give an alternative proof of it based on coding theory.
\begin{corollary}
Let $f$ be an $n$-variable perfect algebraic immune function. Then
$n=2^{\tau}+1$ when $\wt(f)$ is even and $n=2^{\tau}$ when $\wt(f)$ is odd, where $\tau$ is a positive integer.
\end{corollary}
\begin{proof}
Let $f$ be a perfect algebraic immune function and let $D$ be its support.
 Then $\AI (f) \ge n/2$ by Corollary \ref{cor:FAI<=2AI}.
Theorem \ref{thm:PAI-LCD} shows that  $\mathrm{RM}(e,n)^{\overline{D}}$ is an LCD code for any $1\le e \le (n-1)/2$.
Combining Proposition \ref{prop:AI-dim} with Theorem \ref{thm:RM} yields
 $\mathrm{dim}_{\mathbb F_2} \left ( \mathrm{RM}(e,n)^{\overline{D}} \right ) = \sum_{i=0}^{e} \binom{n}{i}$.
 Since $\mathbf{1}_{\wt (f)}$ lies in $\mathrm{RM}(e,n)^{\overline{D}}$,
 $\left ( \mathrm{RM}(e,n)^{\overline{D}} \right )^{\perp}$
 is an even-like LCD code with dimension $\wt (f) - \sum_{i=0}^{e} \binom{n}{i}$.
 We conclude from Proposition \ref{prop:dim-even} that $\wt (f) \equiv  \sum_{i=0}^{e} \binom{n}{i} \pmod{2}$ for any $1\le e \le (n-1)/2$.
 This clearly forces
 \begin{eqnarray}\label{eqn:binom-even}
 n\equiv \wt(f) +1 \pmod{2} \text{ and } \binom{n}{e} \equiv 0 \pmod{2},
 \end{eqnarray}
 where $2\le e \le (n-1)/2$.

Let us first consider the case $\wt (f) \equiv 0 \pmod{2}$. Write $n=2^{\tau}+\sum_{i=0}^{\tau-1} a_i 2^i$.
We have $a_0=1$, because  $n \equiv 1 \pmod{2}$. If there existed an $a_i$ such that $a_i=1$ and $1\le i \le \tau -1$,
we would have $2\le 2^i \le (n-1)/2$ and $\binom{n}{2^i} \equiv 1 \pmod{2}$ by Lucas' Theorem, contrary to (\ref{eqn:binom-even}).
Hence $a_i=0$ and $n=2^{\tau}+1$.

Similar arguments apply to the case $\wt (f) \equiv 1 \pmod{2}$. Then we have $n=2^{\tau}$ in this case.This completes the proof. \qed

\end{proof}

As a corollary of Proposition \ref{prop:AI-dim} and Theorem \ref{thm:PAI-LCD}, we have the following,
which provides a way of constructing LCD codes via perfect algebraic immune function.
\begin{corollary}\label{cor:FAI-LCD}
Let $f$ be an $n$-variable perfect algebraic immune function and let $D$ be its support.
Let $e$ be an integer with $1\le e \le (n-1)/2$.
Then
 $\mathrm{RM}(e,n)^{\overline{D}}$ is an LCD code of dimension $\sum_{i=0}^{e} \binom{n}{i}$.
\end{corollary}

Plugging all the families of perfect algebraic immune functions presented in \cite{CarletFeng2008} and  \cite{Liu2012} into Corollary \ref{cor:FAI-LCD}
 will produce a lot of binary LCD codes.

 \begin{corollary}
 Let $n=2^{\tau}$ or $2^{\tau}+1$. Let $D$ be a subset of $\mathbb F_{2^n}$ given by
 \begin{eqnarray*}
 D= \left \{
 \begin{array}{rl}
 \{\alpha^{\ell}, \alpha^{\ell+1}, \cdots,  \alpha^{\ell+2^{\tau-1}-1} \}, &\text{ if } n=2^{\tau}+1,\\
\\
\{0, \alpha^{\ell}, \alpha^{\ell+1}, \cdots,  \alpha^{\ell+2^{\tau-1}-1} \}, &\text{ if } n=2^{\tau},
 \end{array}
 \right .
 \end{eqnarray*}
 where $\alpha$ is a primitive element of $\mathbb F_{2^n}$ and $\ell$ is an integer.
Then
 $\mathrm{RM}(e,n)^{\overline{D}}$ is an LCD code of dimension $\sum_{i=0}^{e} \binom{n}{i}$ for any $1 \le e \le (n-1)/2$.
 \end{corollary}

\section{Conclusions}

In this paper, we investigated some problems on fast algebraic immunity of Boolean functions and LCD codes. More specifically, we pushed further the general study of the fast algebraic immunity and investigated its behavior in particular for certain families of Boolean functions. We have also introduced the related fast immunity profile and showed that the algebraic immunity and the fast algebraic immunity of a Boolean function can be expressed by means of its immunity profile. In addition, we provided new characterizations of perfect algebraic immune  functions  by means of the LCD-ness of punctured Reed-Muller codes. We also contributed to the current work on binary LCD codes (which are the most important codes regarding its applications in armoring implementations against side-channel attacks and fault non-invasive attacks) by constructing a large class of binary LCD codes from perfect algebraic immune functions. The results show  a novel application of perfect algebraic immune functions in addition to their contribution in symmetric cryptography. This offers a new direction of research in this context.


\end{document}